\documentclass[11pt]{huber_article}

\usepackage{amsthm, amsmath, amsfonts}
\usepackage{booktabs}
\usepackage{dsfont}
\usepackage{hyperref}

\newtheorem{theorem}{Theorem}
\newtheorem{lemma}{Lemma}
\theoremstyle{definition}
\newtheorem{definition}{Definition}

\newcommand{\prob}{\mathbb{P}}
\newcommand{\mean}{\mathbb{E}}
\newcommand{\ind}{\mathds{1}}
\newcommand{\berndist}{\textsf{Bern}}

\newcommand*{\secref}[1]{Section~\ref{#1}}

\newcounter{line}
\newcommand{\gk}{\arabic{line})\stepcounter{line}}

\newenvironment{pseudocode}[1]
    {\par\vspace*{9pt}    \setcounter{line}{1}
    \begin{tabular}{rp{0.5\textwidth}}
        \toprule
          \multicolumn{2}{l}{\textbf{#1}} \\
          \midrule 
    }
    { 
    \bottomrule
    \end{tabular}
    \vspace*{6pt}
    }

\begin{document}

\title{Designing perfect simulation algorithms using local correctness}

\author{Mark Huber  \\ {\tt mhuber@cmc.edu}}

\maketitle

\pagenumbering{gobble}

\begin{abstract}
Consider a randomized algorithm that draws samples exactly from a
distribution using recursion. Such an algorithm is called a perfect
simulation, and here a variety of methods for building this type of
algorithm are shown to derive from the same result: the Fundamental
Theorem of Perfect Simulation (FTPS). The FTPS gives two necessary and
sufficient conditions for the output of a recursive probabilistic
algorithm to come exactly from the desired distribution. First, the
algorithm must terminate with probability 1. Second, the algorithm must
be locally correct, which means that if the recursive calls in the
original algorithm are replaced by oracles that draw from the desired
distribution, then this new algorithm can be proven to be correct. While
it is usually straightforward to verify these conditions, they are
surprisingly powerful, giving the correctness of Acceptance/Rejection,
Coupling from the Past, the Randomness Recycler, Read-once CFTP, Partial
Rejection Sampling, Partially Recursive Acceptance Rejection, and
various Bernoulli Factories. We illustrate the use of this algorithm by
building a new Bernoulli Factory for linear functions that is 41\%
faster than the previous method.
\end{abstract}

\section{Introduction}\label{introduction}

The ability to sample from complicated unnormalized densities is very
beneficial in building \((\epsilon, \delta)\)-randomized approximation
schemes for many \#P complete problems. Examples include approximation
algorithms for the volume of a convex body~\cite{dyerfk1991}, the
permanent of 0-1 matrices~\cite{jerrums1989, huber2006a}, the
normalizing constant of the Ising model~\cite{jerrums1993}, the number
of solutions to a Disjunctive Normal Form logical statement
\cite{karpl1983}, the number of \(s-t\) paths in a network
\cite{provanb1983} and many more.

\begin{definition}
For a nonnegative measurable function $h$, say that $X$ has \emph{unnormalized density $h$ with respect to measure $\mu$} (write $X \sim h$) if 
\[
0 < Z_h = \int_x h(x) \ d\mu < \infty.
\]
and for all measurable sets $A$,
\[
\prob(X \in A) = \frac{\int_{x \in A} h(x) \ d\mu}{Z_h}.
\]
\end{definition}

\pagenumbering{arabic}

For most such problems, exact computation of \(Z_h\) is a \#P complete
problem, as it is for all the examples listed in the first paragraph.
Hence randomized approximation algorithms are widely used to give
estimates of the solution rather than the exact answer.

Given a target distribution \(\pi\), a valuable tool in the construction
of algorithms for drawing a random variate \(X\) exactly from \(\pi\) is
recursion.

\begin{definition}
Say that \textbf{Alg}$(\alpha)$ is a \emph{probabilistic recursive scheme (PRS)} if the algorithm makes random choices, and in the course of running is allowed to call itself recursively, possibly with different values of the parameter $\alpha$.
\end{definition}

As an example, consider the following acceptance/rejection (AR) style
algorithm that draws uniformly from \(\{1,2,3,4,5\}\).

\begin{pseudocode}{AR1$(\alpha)$}
\gk & Repeat \\
\gk & \quad Draw $X$ uniformly from $\{1,2,\ldots,10\}$ \\
\gk & Until $X \in \{1,2,3,4,5\}$ \\
\gk & Return $X$ \\
\end{pseudocode}

\textbf{AR1} uses a repeat loop, but this is equivalent to the following
recursive form.

\begin{pseudocode}{AR2$(\alpha)$}
\gk & Draw $X$ uniformly from $\{1,2,\ldots,10\}$ \\
\gk & If $X \in \{1,2,3,4,5\}$ \\
\gk & \quad Return $X$ \\
\gk & Else \\
\gk & \quad $X \leftarrow \textbf{AR2}(\alpha)$ \\
\gk & \quad Return $X$ \\
\end{pseudocode}

In practice, code is typically faster using the repeat loop formulation
rather than the recursive version. However, from an execution point of
view \textbf{AR1} and \textbf{AR2} are the same algorithm. Hence we will
say that both \textbf{AR1} and \textbf{AR2} are examples of a PRS.

\begin{definition}
Given a parameter set $\mathcal{P}$ and family of target distributions $\{\pi_\alpha:\alpha \in \mathcal{P}\}$, a PRS \textbf{Alg}$(\alpha)$ is a \emph{perfect simulation} or \emph{perfect sampling} algorithm if for all $\alpha$, the algorithm terminates in finite time with probability 1 and returns a draw with distribution equal to $\pi_\alpha$.
\end{definition}

\begin{lemma}
\textbf{AR2} is a perfect simulation algorithm.
\end{lemma}

\begin{proof}
Here $\pi_\alpha$ is the uniform distribution over $\{1,2,3,4,5\}$ for all $\alpha$.  At each step there is a $5/6$ chance the algorithm will terminate without a recursive call, so the algorithm terminates in finite time with probability 1.  Also, for $i \in \{1,\ldots,5\}$, the recursive call is exactly the same as the original call, so the probability the output is $i$ equals the chance that the initial draw of $X$ is $i$ plus the chance that the initial draw of $X$ is 6 times the chances that the recursive call returns $i$.  That is,
\begin{align*}
\prob(X = i) &= (1/6) + (1/6)\prob(X = i),
\end{align*}
which yields $\prob(X = i) = 1/5$ for all $i \in \{1,\ldots,5\}$.
\end{proof}

Here the proof was easy because the recursive call was exactly the same
as the original call. But now consider an \emph{adaptive} AR algorithm
that modifies the random choice based on the input parameter. Suppose
that the input parameter \(\alpha\) is a positive integer.

\begin{pseudocode}{AR3$(\alpha)$}
\gk & Draw $X$ uniformly from $\{1,2,\ldots,\max(5, \alpha)\}$ \\
\gk & If $X \in \{1,2,3,4,5\}$ \\
\gk & \quad Return $X$ \\
\gk & Else \\
\gk & \quad $Y \leftarrow \textbf{AR3}(X)$ \\
\gk & \quad Return $Y$ \\
\end{pseudocode}

Intuitively, this should still return output uniformly distributed over
\(\{1,2,3,4,5\}\). However, this minor change has made the analysis much
more difficult, as the technique used in the last proof no longer
applies since there is no reason initially to believe (for instance)
that \(\textbf{Alg}(100)\) and \(\textbf{Alg}(10)\) have the same output
distribution.

The purpose of this paper is to present a new version of the
\emph{Fundamental Theorem of Perfect Simulation (FTPS)} that gives two
necessary and sufficient conditions for when a PRS is actually a perfect
simulation algorithm. These conditions are typically both easy to
verify. The first condition is that the algorithm must terminate in
finite time with probability 1. The second condition is called
\emph{local correctness}.

\begin{definition}
In a call $\textbf{Alg}(\alpha)$ to a PRS, suppose that any recursive calls of the form $\textbf{Alg}(\beta)$ within the algorithm are replaced with oracles that draw exactly from their target distribution $\pi_\beta$.  If the output of the resulting algorithm with oracles can be shown to come from $\pi_\alpha$ for all $\alpha$, then the algorithm is \emph{locally correct} with respect to $\{\pi_\alpha:\alpha \in \mathcal{P}\}$.
\end{definition}

When these two conditions hold for all input parameters, the output of
the PRS will be exactly from the desired result.

\begin{theorem}[Fundamental Theorem of Perfect Simulation]
Suppose $\textbf{Alg}(\alpha)$ is a PRS that satisfies the following.
\begin{enumerate}
  \item It terminates in finite time with probability 1 for all $\alpha$.
  \item It is locally correct with respect to $\{\pi_\alpha:\alpha \in \mathcal{P}\}$.
\end{enumerate}
Then the output of the algorithm is an exact draw from $\pi_\alpha$.
\end{theorem}

This simple condition is surprisingly powerful. Consider the example
from earlier.

\begin{lemma}
$\textbf{AR3}(\alpha)$ is a perfect simulation algorithm whose output is uniform over $\{1,2,3,4,5\}$ for all positive integers $\alpha$.
\end{lemma}

\begin{proof}
Here $\pi_\alpha$ is uniform over $\{1,2,3,4,5\}$ for all positive integers $\alpha$.  Since the parameter in any recursive call is at most the initial value $\alpha_0$ passed to the algorithm, the probability of termination is at least $5/\alpha_0$ and the algorithm will terminate in finite time with probability 1.

Then to show correctness using the FTPS, we assume that the recursive call at line 5 actually is a draw from the correct target distribution that is uniform over $\{1,\ldots,5\}$.  Let $\alpha_0$ be the initial parameter passed to the algorithm.  If $\alpha_0 \leq 5$ then line 1 just draws uniformly from $1$ to $5$, so assume $\alpha_0 > 5$.  Then for $i \in \{1,\ldots,5\}$, the chance that $Y = i$ is the chance that $X = i$ plus the chance that $X > 5$ times the chance that the oracle returns $Y = i$.  That is,
\[
\prob(Y = i) = \frac{1}{\alpha_0} + \left(1  - \frac{5}{\alpha_0}\right)\frac{1}{5} = \frac{1}{5}.
\]
Hence the algorithm is locally correct, and the FTPS gives that the overall algorithm is correct.
\end{proof}

The value of the FTPS is similar to the use of the Ergodic Theorem in
Markov chain Monte Carlo. Without going into too much detail, that
theorem states that a Markov chain that is irreducible, aperiodic, and
stationary with respect to \(\pi\) will have \(\pi\) as the limiting
distribution regardless of the starting state. Irreducibility and
aperiodicity is similar to our first requirement that the algorithm must
terminate with probability 1. Stationarity is similar to local
correctness in that it is usually not too difficult to verify (although
there are exceptions.)

The simplicity of the Ergodic Theorem has led to a multitude of Markov
chain Monte Carlo algorithms. In the same way, the simplicity of the
local correctness criterion encourages variety of design in perfect
simulation algorithms. In particular, what the FTPS does for us is
three-fold.

\begin{enumerate}
\def\labelenumi{\arabic{enumi}.}
\item
  It gives easy proofs of the most common perfect simulation protocols,
  such as Acceptance/Rejection (AR), Coupling from the Past (CFTP), and
  the Randomness Recycler (RR).
\item
  It allows us to build variants of these algorithms, such as adaptive
  AR, partially recursive AR, time-heterogeneous CFTP, and fractal time
  CFTP, without having to start over from scratch on the proofs. We can
  make these minor changes without worrying that a minor modification of
  the algorithm will change the resulting output.
\item
  It allows the building of more complicated algorithms that use
  recursion in more unusual ways, such as the first polynomial expected
  time Bernoulli Factory \cite{huber2016a}.
\end{enumerate}

The remainder of this paper is organized as follows. In the next section
we give the proof of the FTPS. In \secref{SEC:protocols} we employ the
result to give new proofs of many of the most important protocols in the
construction of perfect simulation algorithms. In particular we present
a new Bernoulli Factory algorithm whose proof of correctness is greatly
eased by using the FTPS. Finally, in \secref{SEC:models} we consider
ramifications of the model of computation that we are using.

\section{Proof of the FTPS}\label{SEC:protocols}

In order to prove the FTPS, we will consider a coupled sequence of
algorithms. Our initial call to \(\textbf{Alg}(\alpha)\) is allowed to
use recursion an unbounded number of times. Suppose that for each \(n\)
we construct a new algorithm \(\textbf{Alg}_n(\alpha)\) in the following
way. First we consider what \emph{level of recursion} we are at in our
algorithm.

\begin{definition}
An initial call to a PRS $\textbf{Alg}$ is said to be at \emph{level of recursion 0.}  For $i$ a positive integer, say the call to $\textbf{Alg}$ occurs at the \emph{level of recursion $i$} if it was called by a call that was at level of recursion $i - 1$.
\end{definition}

Then set up \(\textbf{Alg}_n(\alpha)\) as follows. If the level of
recursion \(\ell\) of a call to \(\textbf{Alg}_n(\alpha)\) is less than
\(n\), then the algorithm behaves exactly like \(\textbf{Alg}(\alpha)\).
However, if the the call is at level of recursion \(\ell = n\), the
algorithm replaces the recursive calls with oracles that generate draws
exactly from the target distribution.

So for example, we might call \(\textbf{AR3}(100)\) and have the
following outcome:

\begin{center}
\begin{tabular}{lllllll}
sample run:         & $\textbf{AR3}(100)$ & $\textbf{AR3}(67)$ & $\textbf{AR3}(11)$ & Outputs $X = 4$\\
level of recursion: & 0 & 1 & 2 &  
\end{tabular}
\end{center}

Throughout this section, we will use \(X\) to denote the output from
\(\textbf{Alg}(\alpha)\), and \(Y_n\) to denote the output of
\(\textbf{Alg}_n(\alpha)\). For an original call,
\(\textbf{Alg}(\alpha)\), let \(T\) denote the supremum of the levels of
recursion employed by the algorithm in the call. By definition then
\(Y_n = X\) for any \(n \geq T\).

In the sample run, since \(\textbf{Alg}(11)\) was the last call at level
of recursion 2, \(T = 2\), and \(X = 4 = Y_2 = Y_3 = \cdots\). The
outputs \(Y_0\) and \(Y_1\) used oracles at level of recursion 0 and 1
respectively to determine their value and so might have different values
than \(X\).

The first step in the proof of the FTPS is to use local correctness to
show that for every \(n\), \(Y_n\) has the correct distribution.

\begin{lemma}
For all nonnegative integers $n$, the output of $\textbf{Alg}_n(\alpha)$ (where $\textbf{Alg}$ is locally correct) is an exact draw from $\pi_\alpha$.
\end{lemma}

\begin{proof}  When $n = 0$, this is just the definition of local correctness.  Suppose it holds for $\textbf{Alg}_n(\alpha)$, and consider $\textbf{Alg}_{n + 1}(\alpha)$.  

For the output of any call at level of recursion $n + 1$, any further recursive calls use oracles.  So local correctness guarantees that the output of level of recursion $n + 1$ comes exactly from the desired distribution.  So it is the same as if that output came from an oracle with the correct distribution. But that means $Y_{n+1}$ has the same distribution as $Y_n$.  That allows us to use the induction hypothesis to say that $Y_{n+1}$ must have the correct distribution, which completes the induction. 
\end{proof}

We are now ready to prove the FTPS.

\begin{proof}[Proof of the FTPS]
Let $A$ be any measurable set, and $n$ be any positive integer.  Then
\begin{align*}
\prob(X \in A) &= \prob(X \in A, T \leq n) + \prob(X \in A, T > n).
\end{align*}
If $T \leq n$, then $X = Y_n$, so 
\begin{align*}
\prob(X \in A) &= \prob(Y_n \in A, T \leq n) + \prob(X \in A, T > n) \\
&= \prob(Y_n \in A) - \prob(Y_n \in A, T > n) + \prob(X \in A, T > n).
\end{align*}
The probability that $Y_n \in A$ is the probability that a draw from $\pi_\alpha$ falls in $A$.  Denote this probability $\pi_\alpha(A)$.  Then
\begin{align*}
|\prob(X \in A) - \pi_\alpha(A)| &=
 |\prob(X \in A) - \prob(Y_n \in A)| \\
 &= |\prob(X \in A, T > n) - \prob(Y_n \in A, T > n)| \\
 &\leq \prob(T > n).
\end{align*}

Since we assumed that $T$ was finite with probability 1, as $n \rightarrow \infty$, $\prob(T > n) \rightarrow 0$, hence $\prob(X \in A) = \pi_\alpha(A)$, and we are done.
\end{proof}

\section{Perfect simulation protocols}\label{SEC:bf}

The two most common perfect simulation protocols are
acceptance/rejection (sometimes called rejection sampling) and Coupling
from the Past, so let us begin with those.

\subsection{Acceptance Rejection}\label{acceptance-rejection}

Consider the following acceptance/rejection algorithm for unnormalized
densities. Suppose that \(g \geq h\) are two unnormalized densities such
that it is easy to draw random variates from \(g\) and we wish to obtain
draws from \(h\). The following algorithm goes back to
\cite{vonneumann1951} but has been extended and used for applications
ranging from counting DNF satisfying assignments \cite{karpl1983} to
drawing exact instances of solutions to Stochastic Differential
Equations \cite{beskospr2006}.

\begin{pseudocode}{AR4$(g)$}
\gk & Draw $X$ using $g$ \\
\gk & Draw $U$ uniformly from $[0,1]$ \\
\gk & If $U < h(X)/g(X)$ \\
\gk & \quad Return $X$ \\
\gk & Else \\
\gk & \quad $Y \leftarrow \textbf{AR4}(g)$ \\
\gk & \quad Return $Y$ \\
\end{pseudocode}

\begin{lemma}
For $h$ and $g$ unnormalized densities with $g \geq h$, the output of \textbf{AR4} is $X \sim h$.
\end{lemma}

\begin{proof}
By the FTPS we need termination in finite time and local correctness.  The chance of terminating at each step is 
\begin{align*}
\prob(U \leq h(X)/g(X)) &= \mean[\ind(U \leq h(X)/g(X))] \\
&= \mean[\mean[\ind(U \leq h(X)/g(X))|X]] \\
&= \mean[h(X)/g(X)] \\
&= \int_{x} (h(x)/g(x))(g(x)/Z_g) \ d \mu \\
&= \int_{x} h(x)/Z_g \ d \mu = Z_h/Z_g > 0
\end{align*}
Here $\ind(p)$ is the usual indicator function that is 1 if its argument is true and 0 otherwise.  Since the chance of terminating at each step is a fixed positive number, the chance of running infinitely often is 0.

Now for local correctness.  Let $A$ be a set measurable with respect to $\mu$.  Then the chance that the output lies in $A$ is the chance that the initial draw of $X \in A$ times the chance we accept plus the chance that we reject the initial draw and the recursive call $Y$ falls in $A$ once we substitute in the true distribution for the recursive call.  In notation, if we call the output $W$,
\begin{align*}
\prob(W \in A) &= \prob(X \in A, U \leq h(X)/g(X)) + \prob(U > g(X)/h(X))\prob(Y \in A) \\
&= \frac{\int_{x \in A} g(x) \frac{h(x)}{g(x)} \ d\mu}{Z_g}+ \left(1 - \frac{Z_h}{Z_g}\right)\prob(Y \in A) \\
&= \frac{\int_{x \in A} h(x) \ d\mu}{Z_h} \cdot \frac{Z_h}{Z_g} + \left(1 - \frac{Z_h}{Z_g}\right)\prob(Y \in A) \\ 
&= \prob(Y \in A) \frac{Z_h}{Z_g} + \left(1 - \frac{Z_h}{Z_g}\right)\prob(Y \in A) \\
&= \prob(Y \in A),
\end{align*}
and so the output has the correct distribution.
\end{proof}

We call \(g\) the \emph{envelope density}. In adaptive AR, the envelope
density is refined based on the value of \(X\). For example, the
approach for Log-concave density functions of Wild and Gilks
\cite{wildg1993} works like this. In general, the envelope
function is itself modified as a function of the rejection sample.

\begin{pseudocode}{AR5$(g)$}
\gk & Draw $X$ using $g$ \\
\gk & Draw $U$ uniformly from $[0,1]$ \\
\gk & If $U < h(X)/g(X)$ \\
\gk & \quad Return $X$ \\
\gk & Else \\
\gk & \quad Let $g_a \leq g$ be a new envelope function for $h$ that depends on $X$ \\
\gk & \quad $Y \leftarrow \textbf{AR5}(g_a)$ \\
\gk & \quad Return $Y$ \\
\end{pseudocode}

\begin{lemma}
Suppose for each $g$ and rejected $X$, $g_a$ is still an envelope density of $h$.  Then the output of \textbf{AR4} is $X \sim h$.
\end{lemma}

\begin{proof}
At each recursive call the chance of terminating can only stay the same or increase, and so the chance that we never terminate is still bounded above by 0.

For local correctness, we substitute the recursive call with a direct oracle, and so the change in the argument for the recursive call does not matter.  The same proof as in the previous lemma still works.
\end{proof}

\subsection{Coupling from the Past}\label{coupling-from-the-past}

Coupling from the Past (CFTP) was invented by Propp and Wilson
\cite{proppw1996} as a way to turn approximate sampling by a
Markov chain into perfect simulation. It has since become a mainstay of
perfect simulation algorithms. There have been many variants and
extensions, for instance Read-Once CFTP \cite{wilson2000}, clan of
ancestors \cite{ferrarifg2002}, Bounding chain CFTP \cite{huber2004a},
time-heterogeneous CFTP \cite{huber2006b}, and fractal time CFTP
\cite{huber2008b}.

Rather than repeat the entire evolution of the idea, here we move
straight to the most general version, which revolves around the notion
of a \emph{stationary update function}.

\begin{definition}
Say that $\phi:\Omega \times \mathcal{R} \rightarrow \Omega$ is a \emph{stationary update function with respect to $\pi$} if there is a probability distribution $\prob$ over $\mathcal{R}$ such that if $R \sim \prob$ and $X \sim \pi$ then $\phi(X,R) \sim \pi$ as well.
\end{definition}

We refer to \(X\) as the \emph{state} and \(R\) as the
\emph{random choices} that update the state. For instance, \(\phi\)
might encode taking one or more steps in a Metropolis-Hastings Markov
chain whose stationary distribution is \(\pi\).

For some random choices, the current state might be immaterial in
determining the next state. That is, suppose for all \(x \in \Omega\),
\(\phi(x,R) = y\). Then in this case we say the state has
\emph{completely coupled}.

With this in mind, we present here a very general version of CFTP. For a
given parameter \(\alpha\), suppose we use stationary update function
\(\phi_\alpha\). Moreover, let \(A_\alpha\) be any set such that \[
(\forall r \in A_\alpha)(\forall x_0, x_1 \in \Omega)(\phi_\alpha(x_0,r) = \phi_\alpha(x_1,r)).
\]

\begin{pseudocode}{CFTP$(\alpha)$}
\gk & Draw $R$ using $\prob_{\alpha}$ \\
\gk & If $R \in A_\alpha$  \\
\gk & \quad Let $x$ be any state \\
\gk & \quad Return $\phi_\alpha(x,R)$ \\
\gk & Else \\
\gk & \quad $Y \leftarrow \textbf{CFTP}(f(\alpha))$ \\
\gk & \quad Return $\phi_\alpha(Y,R)$ \\
\end{pseudocode}

\begin{lemma}  Suppose $\phi_\alpha$ is stationary with respect to $\pi$ for all parameters $\alpha$.
If a call to $\textbf{CFTP}(\alpha)$ terminates with probability 1, then the output comes from $\pi$.
\end{lemma}

\begin{proof}
Since we are assuming finite termination, we need only verify local correctness.  Let $A$ be a $\pi$-measurable set.  Let $W$ be the output of $\textbf{CFTP}(\alpha)$.  Let $x$ be any element of $\Omega$.
\begin{align*}
\prob(W \in A) &= \prob(\phi_\alpha(x,R) \in A, R \in A_\alpha) + \prob(\phi_\alpha(Y,R) \in A, R \notin A_\alpha) \\
&= \prob(\phi_\alpha(Y,R) \in A, R \in A_\alpha) + \prob(\phi_\alpha(Y,R) \in A, R \notin A_\alpha) \\
&= \prob(\phi_\alpha(Y,R) \in A) \\
&= \pi(A),
\end{align*}
since $\phi_\alpha$ is a stationary update function.
\end{proof}

\subsection{Bernoulli Factories}\label{bernoulli-factories}

Bernoulli factories were introduced by Asmussen et. al.
\cite{asmussengt1992} as part of an algorithm for generating
samples exactly from the stationary distribution of a regenerative
Markov process. The idea is as follows.

Suppose you have access to a stream of independent, identically
distributed (iid) Bernoulli random variables \(X_1,X_2,\ldots\) with
unknown parameter \(p\). So \(\prob(X_i = 1) = p\),
\(\prob(X_i = 0) = 1 - p\). Write \(X_i \sim \berndist(p)\). Now suppose
that we wish to construct a Bernoulli random variable with a parameter
that is a function of \(p\) using a random number of \(\{X_i\}\). That
is, we want \(Y \sim \berndist(f(p))\) for a known function \(p\) where
\(Y = f_{\text{BF}}(X_1,\ldots,X_T)\) for some function
\(f_{\text{BF}}\) and stopping time \(T\). Such an algorithm is a
\emph{Bernoulli factory}.

For instance, to draw \(Y \sim \berndist(p(1-p))\), just let
\(Y = X_1(1 - X_2)\). For the application in \cite{asmussengt1992},
they needed to be able to draw \(W \sim \berndist(Cp)\) where \(C\) is a
known constant. This problem, although simple to state, turns out to be
surprisingly difficult. Nacu and Peres \cite{nacup2005} showed
that for any function that is analytic and bounded away from 1, it
suffices to have a Bernoulli factory for \(2p\).

The first polynomial expected time algorithm for this problem was given
in \cite{huber2016a}. Here we present a new version of the algorithm
that is slightly simpler to implement and analyze and takes advantage of
the FTPS.

To construct this linear factory, it helps to have a factory for
\(Cp / (1 + Cp)\) available.

\begin{pseudocode}{BF1$(C)$}
\gk & Draw $B$ as $\berndist(C/[1 + C])$ \\
\gk & If $B = 0$ \\
\gk & \quad Return 0 \\
\gk & Else \\
\gk & \quad Draw $X$ as $\berndist(p)$ \\
\gk & \quad If $X = 1$ return $1$ \\
\gk & \quad Else \\
\gk & \quad \quad Return \textbf{BF1}$(C)$ \\
\end{pseudocode}

\begin{lemma}
The output of \textbf{BF1}$(C)$ is $\berndist(Cp/(1+Cp))$.
\end{lemma}

\begin{proof}
If $B = 0$ the algorithm does not call itself recursively, therefore the chance that the algorithm does not terminate in finite time is 0.

Let $W$ be the output of the algorithm.  To show local correctness, we assume the final call in line 8 returns $Y \sim \berndist(Cp/(1 + Cp))$.  Then $W$ must be in $\{0,1\}$, and 
\begin{align*}
\prob(W = 1) &= \prob(B = 1, X = 1) + \prob(B = 1, X = 0, Y = 1) \\
&= \prob(B = 1)[\prob(X = 1) + \prob(X = 0)\prob(Y = 1)] \\
&= (C/(1+C))[p + (1 - p) Cp / (1 + Cp)] = Cp / (1 + Cp),
\end{align*}
giving local correctness.
\end{proof}

Rather than just build a \(Cp\) Bernoulli factory, we will build a
\((Cp)^i\) Bernoulli factory where \(C \geq 1\) and
\(i \in \{0, 1, \ldots\}\). For this factory, we suppose that we know
\(\epsilon > 0\) such that \(Cp \leq 1 - \epsilon\). If no such
\(\epsilon\) exists, then it is impossible to build such an algorithm,
see \cite{huber2016a} for the proof.

\begin{pseudocode}{BF2$(C, i, \epsilon)$}
\gk & If $i = 0$ \\
\gk & \quad Return 1 \\
\gk & If $i > 3.55 / \epsilon$ \\
\gk & \quad Let $\beta = (1 - \epsilon / 2) / (1 - \epsilon)$ \\ 
\gk & \quad Draw $B_1$ as $\berndist(\beta^{-i})$ \\
\gk & \quad If $B_1 = 0$, return 0 \\
\gk & \quad Else return $\textbf{BF2}(\beta C, i, \epsilon/2)$ \\
\gk & Else \\
\gk & \quad Let $B_2 \leftarrow \textbf{BF1}(C)$ \\
\gk & \quad Return $\textbf{BF2}(C, i + 1 - 2 B_2, \epsilon)$ \\
\end{pseudocode}

\begin{lemma}
For $Cp < 1 - \epsilon$, $i \in \{0, 1, 2,\ldots\}$, the output of $\textbf{BF2}(C, i, \epsilon)$ is $\berndist([Cp]^i)$.
\end{lemma}

\begin{proof}
Let $W$ denote the output of the call to the algorithm.
When $i = 0$ the result is true since $W = 1$.  If $i > 3.55 / \epsilon$, then to show local correctness, consider $\prob(W = 1)$.

If $Cp < 1 - \epsilon$, then $C \beta p = C(1-\epsilon/2)p/(1 - \epsilon) \leq 1 - \epsilon / 2$.  So we assume the recursive call in line 6 comes from an oracle, and hence returns $\berndist([\beta C p]^i) \sim \berndist(\beta^{i} [Cp]^i)$.  Then the probability $W = 1$ is $(\beta^{-i})(\beta^{i}[Cp]^i) = [Cp]^i$.

Now suppose $i > 0$ and $i < 3.55 / \epsilon$.  Then $W$ is the result of line 10, which depends on $B_2$.  So
\begin{align*}
\prob(W = 1) &= \prob(B_2 = 1)(Cp)^{i - 1} + \prob(B_2 = 0)(Cp)^{i + 1} \\
&= (Cp)^{i-1}\left[\frac{Cp}{1 + Cp} + \frac{1}{1 + Cp} (Cp)^2 \right] = (Cp)^i.
\end{align*}
Hence the algorithm is locally correct.  

If $i > 3.55 / \epsilon$, then $\beta^{-i} \leq \exp(-3.55 / 2)$.  So every time that $i$ hits 0 or exceeds $3.55 / \epsilon$ there is a $1 - \exp(-3.55 / 2)$ chance of stopping.  When $i$ falls inside those extremes, there is at least a $1/2$ chance that $i$ increases by 1, hence there will be a finite number of steps with probability 1 until $i$ reaches 0 or exceeds $3.55 / \epsilon$, giving a finite number of steps with probability 1 until termination.
\end{proof}

See the appendix for an evaluation of the expected running time.

\subsection{Other protocols}\label{other-protocols}

There are many other protocols for building perfect simulation
algorithms that can employ the FTPS. For instance, the
\emph{design property} of the Randomness Recycler
\cite{huber2000b, huber2015b} gives local correctness for that
approach. Sink popping algorithms \cite{proppw1998, cohnpp2002} can
also be written recursively, which makes their proof of correctness
amenable to use of the FTPS.

Similarly, Partial Rejection Sampling \cite{guojl2017} can be
implemented recursively, allowing the FTPS to be used to show
correctness.

\section{Models of Computation}\label{SEC:models}

Until now we have been looking at these algorithms through the lens of
floating point operations that can be carried out exactly. What changes
when we move to a Turing machine model of computation?

When dealing with randomized algorithms, we can use a Probabilistic
Turing Machine (PTM), which is essentially a Turing Machine that can
access a tape with an iid stream of \(\berndist(1/2)\) random variates.
In this framework, we cannot even simulate exactly a single \(U\)
uniform over \([0,1]\). Instead, for any possible error \(\delta\), we
can compute a result that falls within distance \(\delta\) of the actual
random variable.

When we call the algorithm for a PTM, we must specify the error
tolerance that we are willing to accept. Presumably we are only using
Turing computable functions. Therefore, for a potential recursive call,
it is possible to determine what the error for the recursive call should
be in order to obtain the desired error for the original call. The FTPS
guarantees that the exact algorithm has the correct distribution, and so
the actual computed output will come within the target error of a truly
exact draw.

Of course, none of this is done in practice since usually machine
epsilon is typically small enough to give the desired error for the
original call.

\section{\texorpdfstring{Appendix: Expected running time of
\(\textbf{BF2}\)}{Appendix: Expected running time of \textbackslash{}textbf\{BF2\}}}\label{appendix-expected-running-time-of-textbfbf2}

We first consider \(\textbf{BF1}\).

\begin{lemma}
The expected number of Bernoulli draws needed for $\textbf{BF1}$ is $C / (1 + Cp)$.
\end{lemma}

\begin{proof}
Let $T$ be the number of Bernoulli draws used by the algorithm.  Each call to the algorithm uses either 0 or 1 draws, and then either calls itself recursively or does not.  At each step there is at least a $1/(1+C)$ chance of not recursively calling the algorithm, and so $\mean[T]$ is finite.  This gives rise to the following recursion
\[
\mean[T] = \frac{C}{1 + C}\left[1 + (1 - p) \mean[T] \right],
\]
which solves to $\mean[T] = C / (1 + Cp)$.
\end{proof}

When we call \(\textbf{BF1}\) from \(\textbf{BF2}\) we use the result to
change the number of \(Cp\) coins we need to flip, either increasing by
one with probability \(1/(1+Cp)\) or decreasing it by one with
probability \(Cp/(1+Cp)\). This is known as the
\emph{gambler's ruin problem}. It has been well studied (see, for
instance, \cite{resnick1992}.)

The gambler's ruin problem is a Markov chain where, the state moves from
\(i\) to \(i + 1\) with probability \(r\), and to \(i - 1\) with
probability \(1-r\). The following result is well-known.

\begin{lemma}
Let the initial state be $1$ and suppose the gambler's ruin Markov chain stops when the state reaches $0$ or $n$.  Let $T$ denote the number of steps that are taken until this is reached.  Then for $r > 1/2$,
\[
\mean[T] = \frac{n}{2r - 1} \cdot \frac{1 - \left((1 - r) / r\right)}{1 - \left((1 - r) / r\right)^n} - \frac{1}{2r - 1}.
\]
Also, for any starting state $i \in \{1, 2, \ldots, n\}$,
\[
\mean[T] \leq \frac{n - i}{2r - 1}.
\]
\end{lemma}

Using this, we can upper bound the running time of
\(\textbf{BF2}(C, 1, \epsilon)\).

\begin{lemma}
Let $T$ be the number of Bernoulli draws in $\textbf{BF2}(C, 1, \epsilon)$.  Then
\[
\mean[T] \leq \frac{C(1 + \epsilon^{-1})}{(1 - \exp(-3.55))(1 - 2\exp(-3.55/2))} \]
\end{lemma}

\begin{proof}
Here $r = 1 / (1 + Cp)$, so $1/(2r - 1) = (1 - Cp) / (1 + Cp)$ and $(1 - r) / r = Cp$.
Initially we begin with $(Cp)^i$ and stop when $i = \lceil 3.55/\epsilon \rceil$.  So from the previous lemma the number of draws $T_0$ obeys
\[
\mean[T_0] \leq \left\lceil \frac{3.55}{\epsilon} \right\rceil \cdot \frac{1 + Cp}{1 - Cp} \cdot \frac{1 - Cp}{1 - (Cp)^{3.55/\epsilon}} \cdot \frac{C}{1 + Cp} = \left \lceil \frac{3.55}{\epsilon} \right \rceil \cdot \frac{C}{1 - (Cp)^{3.55/\epsilon}}. 
\]
Using $1 + x \leq \exp(x)$ and $Cp \leq 1 - \epsilon$ gives
\[
\mean[T_0] \leq \left\lceil \frac{3.55}{\epsilon} \right\rceil \frac{C}{1 - \exp(-3.55)}.
\]

If we fail to exit then we change from $i = \lceil 3.55 / \epsilon \rceil$ to $\lceil 3.55 \cdot 2 / \epsilon \rceil$, and so on.  Consider the chance that we need to make draws that move our problem from $i = \lceil 3.55 \cdot 2^k/\epsilon \rceil$ to $i = \lceil 3.55 \cdot 2^{k+1}/\epsilon \rceil$.  The difference between these two numbers is at most $\lceil 3.55 \cdot 2^k/\epsilon \rceil$. So the expected number of steps needed by the previous lemma is
\[
\mean[T_k] \leq \lceil 3.55 \cdot 2^k/\epsilon \rceil \frac{1 - Cp}{1 + Cp} \cdot \frac{C}{1 - Cp} = \frac{\lceil 3.55 \cdot 2^k/\epsilon \rceil C}{1 - Cp}.
\]

So now consider the question of the probability that we need to make these moves in the first place.  Let $q$ denote the probability that in the initial run, $i > 2/\epsilon$ is reached before $i = 0$.  Then by local correctness,
\[
Cp = (1 - q) \cdot 1 + q \cdot (Cp)^i \Rightarrow q = (1 - Cp) / (1 - (Cp)^i) \leq 
  (1 - Cp) / (1 - \exp(-3.55)).
\]

Furthermore, there is only a $\beta^{-i}$ chance that we do not return a 0 and enter the $k$th level of recursion given that we are already at the $k - 1$ level of recursion.  Note
\[
\beta^{-i} \leq \left( \frac{1 - \epsilon / 2}{1 - \epsilon} \right)^{-3.55/\epsilon} \leq \exp(-3.55/2)
\]

Combining, the chance of entering level of recursion $k$ is at most
\[
\exp(-3.55 k / 2) (1 - Cp) / (1 - \exp(-3.55))
\]

This makes (by the Monotone Convergence Theorem) the upper bound on the expected number of steps
\[
\mean[T] \leq \left \lceil \frac{3.55}{\epsilon} \right \rceil \frac{C}{1 - \exp(-3.55)} + \sum_{k=1}^\infty \exp(-3.55 k / 2) \left \lceil \frac{3.55\cdot 2^k}{\epsilon} \right \rceil \frac{C}{1 - \exp(-3.55)}
 \]
 Using $\lceil x \rceil \leq x + 1$ together with standard facts about geometric series then completes the proof.
\end{proof}

The constant of \(3.55\) was found by numerically optimizing the
constant in front of the \(C(1 + \epsilon^{-1})\) term to minimize it.
In \cite{huber2016a}, it was shown that any algorithm of this type must
take \(\Omega(\epsilon^{-1} C)\) steps on average and so the best that
we can hope for is to improve the constant on the bound on the running
time. Note that \([(1-\exp(-3.55))(1-2\exp(-3.55/2))] \leq 5.53\). The
algorithm presented in \cite{huber2016a} was shown to take at most
\(9.5 \epsilon^{-1}C\) so this represents a marked improvement in the
constant factor.

\section{Acknowledgements}

This work supported by U.S. National Science Foundation grant DMS-1418495.  

\renewcommand\refname{References}

%\bibliographystyle{plain}
%\bibliography{refs}

\begin{thebibliography}{10}

\bibitem{asmussengt1992}
S.~Asmussen, P.~W. Glynn, and H.~Thorisson.
\newblock Stationarity detection in the initial transient problem.
\newblock {\em ACM Trans. Modeling and Computer Simulation}, 2(2):130--157,
  1992.

\bibitem{beskospr2006}
A.~Beskos, O.~Papspiliopoulous, and G.~O. Roberts.
\newblock Retrospective exact simulation of diffusion sample paths with
  applications.
\newblock {\em Bernoulli}, 12(6):1077--1098, 2006.

\bibitem{cohnpp2002}
H.~Cohn, R.~Pemantle, and J.~Propp.
\newblock Generating a random sink-free orientation in quadratic time.
\newblock {\em Electron. J. Combin.}, 9(1), 2002.

\bibitem{dyerfk1991}
M.~Dyer, A.~Frieze, and R.~Kannan.
\newblock A random polynomial-time algorithm for approximating the volume of
  convex bodies.
\newblock {\em J. Assoc. Comput. Mach.}, 38(1):1--17, 1991.

\bibitem{ferrarifg2002}
P.A. Ferrari, R.~Fern\'andez, and N.~L. Garcia.
\newblock Perfect simulation for interacting point processes, loss networks and
  ising models.
\newblock {\em Stochastic Process. Appl.}, 102(1):63--68, 2002.

\bibitem{huber2000b}
J.~A. Fill and M.~L. Huber.
\newblock The {R}andomness {R}ecycler: A new approach to perfect sampling.
\newblock In {\em Proc. 41st Sympos. on Foundations of Comp. Sci.}, pages
  503--511, 2000.

\bibitem{guojl2017}
H.~Guo, M.~Jerrum, and J.~Liu.
\newblock Uniform sampling through the {L}ovasz local lemma.
\newblock In {\em STOC}, 2017.

\bibitem{huber2004a}
M.~Huber.
\newblock Perfect sampling using bounding chains.
\newblock {\em Annals of Applied Probability}, 14(2):734--753, 2004.

\bibitem{huber2006a}
M.~Huber.
\newblock Exact sampling from perfect matchings of dense regular bipartite
  graphs.
\newblock {\em Algorithmica}, 44:183--193, 2006.

\bibitem{huber2006b}
M.~Huber.
\newblock Fast perfect sampling from linear extensions.
\newblock {\em Discrete Mathematics}, 306:420--428, 2006.

\bibitem{huber2008b}
M.~Huber.
\newblock Perfect simulation with exponential tails.
\newblock {\em Random Structures Algorithms}, 33(1):29--43, 2008.

\bibitem{huber2016a}
M.~Huber.
\newblock Nearly optimal {B}ernoulli factories for linear functions.
\newblock {\em Combin. Probab. Comput.}, 25(4):577--591, 2016.
\newblock {a}rXiv:1308.1562.

\bibitem{huber2015b}
M.~L. Huber.
\newblock {\em Perfect {S}imulation}.
\newblock Number 148 in Chapman \& Hall/CRC Monographs on Statistics \& Applied
  Probability. CRC Press, 2015.

\bibitem{jerrums1989}
M.~Jerrum and A.~Sinclair.
\newblock Approximating the permanent.
\newblock {\em J. Comput.}, 18:1149--1178, 1989.

\bibitem{jerrums1993}
M.~Jerrum and A.~Sinclair.
\newblock Polynomial-time approximation algorithms for the {I}sing model.
\newblock {\em SIAM J. Comput.}, 22:1087--1116, 1993.

\bibitem{karpl1983}
R.~M. Karp and M.~Luby.
\newblock Monte-{C}arlo algorithms for enumerating and reliability problems.
\newblock In {\em Proc. FOCS}, pages 56--64, 1983.

\bibitem{nacup2005}
S.~Nacu and Y.~Peres.
\newblock Fast simulation of new coins from old.
\newblock {\em Ann. Appl. Probab.}, 15(1A):93--115, 2005.

\bibitem{proppw1996}
J.~G. Propp and D.~B. Wilson.
\newblock Exact sampling with coupled {M}arkov chains and applications to
  statistical mechanics.
\newblock {\em Random Structures Algorithms}, 9(1--2):223--252, 1996.

\bibitem{proppw1998}
J.~G. Propp and D.~B. Wilson.
\newblock How to get a perfectly random sample from a generic {M}arkov chain
  and generate a random spanning tree of a directed graph.
\newblock {\em J. Algorithms}, 217:170--217, 1998.

\bibitem{provanb1983}
J.~S. Provan and M.~O. Ball.
\newblock The complexity of counting cuts and of computing the probability that
  a graph is connected.
\newblock {\em SIAM J. Comput.}, 12:777--788, 1983.

\bibitem{resnick1992}
S.~Resnick.
\newblock {\em Adventures in Stochastic Processes}.
\newblock Birkh\"{a}user, 1992.

\bibitem{vonneumann1951}
J.~von Neumann.
\newblock Various techniques used in connection with random digits.
\newblock In {\em Monte Carlo Method}, Applied Mathematics Series 12,
  Washington, D.C., 1951. National Bureau of Standards.

\bibitem{wildg1993}
P.~Wild and W.~R. Gilks.
\newblock Algorithm as 287: Adaptive rejection sampling from log-concave
  density functions.
\newblock {\em J.R. Stat. Soc. Ser. C. Appl. Stat.}, 42(4):701--709, 1993.

\bibitem{wilson2000}
D.~B. Wilson.
\newblock How to couple from the past using a read-once source of randomness.
\newblock {\em Random Structures Algorithms}, 16(1):85--113, 2000.

\end{thebibliography}

\end{document}